\pdfoutput=1
\documentclass[11pt]{article}

\usepackage{arxiv}

\usepackage{graphicx}
\usepackage{amsmath}
\usepackage{amssymb}
\usepackage{amsthm}
\usepackage{booktabs}
\usepackage{multirow}
\usepackage{microtype}
\usepackage[numbers]{natbib}
\usepackage{xcolor}
\usepackage{url}
\usepackage{hyperref}

\newtheorem{definition}{Definition}
\newtheorem{lemma}{Lemma}
\hypersetup{
  colorlinks=true,
  linkcolor=blue,
  citecolor=blue,
  urlcolor=blue
}

\title{Agentproof: Static Verification of Agent Workflow Graphs}

\author{
  Melwin Xavier \\
  Lule\aa{} tekniska universitet, Sweden \\
  \texttt{melwin.xavier@ltu.se}
  \and
  Vaisakh M A \\
  Independent Researcher \\
  \texttt{997vaisakh@gmail.com}
  \and
  Melveena Jolly \\
  Independent Researcher \\
  \texttt{melveenajollyk@gmail.com}
  \and
  Midhun Xavier \\
  Independent Researcher \\
  \texttt{midhun@industriagents.com}
}

\date{}

\begin{document}
\maketitle

\begin{abstract}
Agent frameworks increasingly encode tool-using behavior as explicit
workflow graphs, yet safety enforcement remains a runtime concern.
These frameworks expose analyzable graph structure
through their APIs, enabling pre-deployment static verification of
safety properties that runtime guardrails can only check reactively.

This paper presents \textsc{Agentproof}, a system that automatically extracts a
unified abstract graph model from four major agent frameworks
(LangGraph, CrewAI, AutoGen, Google ADK), applies six structural
checks with witness trace generation, and evaluates temporal safety
policies via a DSL compiled to deterministic finite automata---both
statically through a graph\,$\times$\,DFA product construction and at
runtime over event traces. Unlike
general-purpose model checkers, Agentproof requires no manual modeling.

In a curated benchmark of 18 author-constructed workflows, 27\% of the
benchmark contain structural defects (dead-end nodes, unreachable exits)
and 55\% violate a human-gate policy when enforced---distinct categories
that prior work conflates. All 15 temporal policies defined fit within
the seven-form DSL fragment, and verification completes in sub-second
time for graphs up to 5{,}000 nodes. The corpus serves as a reproducible
benchmark for evaluating static verification tools rather than as a
prevalence study; defect rates reflect tool detection capability on a
targeted benchmark, not base rates in production systems. Nonetheless,
static graph verification complements runtime guardrails by catching
topology-level defects that runtime tools miss unless the offending path
is exercised.
\end{abstract}

\section{Introduction}
\label{sec:introduction}

Large language models (LLMs) are increasingly deployed as autonomous
agents capable of interacting with external tools and
APIs~\citep{schick2024toolformer, yao2023react}. Rather than generating
text in isolation, these agents retrieve documents, query databases,
send messages, and trigger downstream workflow actions with
real-world side effects. A growing body of work surveys this
paradigm~\citep{wang2024agent_survey}, and several orchestration
frameworks have emerged to structure agent behavior as explicit workflow
graphs: LangGraph~\citep{langgraph_docs},
CrewAI~\citep{crewai_docs}, AutoGen~\citep{wu2023autogen}, and Google
ADK~\citep{google_adk_docs} each represent computation steps as nodes
and allowable transitions as edges.

Despite this rich structural representation, safety enforcement in
current agent systems remains predominantly a \emph{runtime} concern.
Tools such as NeMo Guardrails~\citep{nemo_guardrails} and
LlamaGuard~\citep{llamaguard} intercept individual LLM calls or tool
invocations and apply content-level or policy-level filters at execution
time. While valuable for catching toxic outputs and prompt injection
attempts, these runtime approaches have three fundamental limitations:
(i)~they impose per-call latency overhead, (ii)~they detect violations
only when the offending execution path is actually exercised, and
(iii)~they provide no coverage guarantees over paths that remain
untriggered during testing.

\paragraph{Motivating example.}
Consider an email triage agent with the following workflow: incoming
emails are classified by intent, routed to either an urgent or normal
handler, drafted, and sent. During iterative development, a developer
adds a \texttt{draft\_response} node on the normal-priority path but
forgets to connect it to the \texttt{send} node. The result is a
\emph{dead end}: normal-priority emails enter the draft stage and
silently halt. A runtime guardrail would not catch this defect unless
the normal-priority path is exercised with an email that triggers it.
A static analysis of the workflow graph, however, immediately flags the
dead-end node and produces a witness trace:
\texttt{\_\_start\_\_ $\to$ classify $\to$ router $\to$
normal\_handler $\to$ draft\_response} (stuck).

\paragraph{Research gap.}
This class of defect---a topological error in the workflow
graph cannot be addressed by content-level runtime guardrails. In
principle, classical model checking techniques could detect such
errors: reachability analysis, deadlock detection, and temporal logic
verification are well-established in the formal methods
literature~\citep{clarke1999modelchecking, baier2008modelchecking}.
However, general-purpose model checkers such as
SPIN~\citep{holzmann1997spin} and NuSMV~\citep{cimatti2002nusmv}
require manual translation of the system under analysis into a
dedicated modeling language---a prohibitive overhead for agent
developers iterating on workflow designs. Similarly, the business
process management (BPM) community has a long tradition of workflow
soundness verification~\citep{vanderaalst2011bpm}, but these methods
target visual modeling tools and standardized notations (e.g., BPMN,
Petri nets), not the programmatic API objects through which agent
frameworks define their graphs. To date, no system automatically
extracts and statically verifies agent workflow graphs directly from
framework source code.

\paragraph{Proposed approach.}
This paper argues for \emph{pre-deployment static verification} of
agent workflow graphs. The key observation is that modern agent
frameworks expose workflow structure through their APIs, making the
graph amenable to automated extraction and analysis. When agent behavior
is graph-structured, safety properties reduce to reachability,
isolation, and temporal constraints over the topology properties that
can be checked exhaustively without running the
agent~\citep{clarke1999modelchecking}.

This paper presents Agentproof, a practical system for static analysis
of agent workflows that focuses on properties (a)~expressible over the
workflow graph and (b)~independent of LLM text generation semantics.
Unlike general-purpose model checkers, Agentproof extracts the analysis
model automatically from framework source code, requiring no manual
modeling effort.

\paragraph{Contributions.}
\begin{itemize}
  \item A cross-framework extraction pipeline: extractors for four major
  agent frameworks (LangGraph, CrewAI, AutoGen, Google ADK) automatically
  bridge heterogeneous representations---each with different entry/exit
  conventions, edge semantics, and hierarchy models---into a unified
  abstract workflow model with typed nodes and edges, eliminating the
  manual modeling step required by general-purpose model
  checkers~(Section~\ref{sec:graph_model}).
  \item A pre-deployment verification pipeline comprising six
  structural checks with \emph{witness trace generation} and a temporal
  policy DSL covering the safety fragment of LTL, compiled to
  deterministic finite automata~(Section~\ref{sec:verification_methods}).
  \item An empirical evaluation on 18 curated workflows demonstrating
  that structural defects (dead-end nodes, unreachable exits) and
  policy violations (missing human gates) are common, with sub-second
  verification for graphs up to 5{,}000
  nodes~(Section~\ref{sec:evaluation}).
\end{itemize}

\paragraph{Scope.}
The proposed approach validates workflow design decisions early: ``can this
graph reach a destructive tool?'' or ``must human review occur between
two sensitive actions?'' It does not attempt to prove semantic
properties of LLM outputs (e.g., truthfulness), which remain
inherently non-deterministic. This boundary is discussed precisely in
the threat model (Section~\ref{sec:threat_model}).

\section{Background}
\label{sec:background}

\subsection{Agent workflow graphs}
Many agent frameworks represent behavior as a directed graph, sometimes
called a \emph{state graph} or \emph{workflow graph}. Nodes correspond
to computation steps: an LLM call, a tool invocation, a
router/selector that branches based on state, or a human approval gate.
Edges correspond to transitions between steps, including conditional
branches, parallel fan-outs, and loop back-edges.

This explicit control-flow model is a key difference between
graph-orchestrated agents and monolithic ``single prompt''
applications: the graph structure can be analyzed statically without
executing the agent.

\subsection{Structural verification}
Structural properties depend only on graph topology and type
annotations. Typical examples include:
(i)~reachability (can a node be reached from the entry point?),
(ii)~dead ends (nodes with no outgoing transitions),
(iii)~isolation (must a sensitive tool be preceded by a gate node?),
and (iv)~sanity constraints on router nodes.

These checks are inexpensive: they can be implemented with standard
graph algorithms (BFS/DFS) and run in time linear in $|V| + |E|$.

\subsection{Temporal monitors}
Temporal logics such as linear temporal logic
(LTL)~\citep{pnueli1977temporal} provide a formal language to specify
constraints over execution traces. A common verification strategy is to
compile temporal constraints into automata and reason over traces or
over the product of a system model and the
automaton~\citep{vardi1994automata}.

Agentproof compiles a temporal policy DSL into deterministic finite
automata and supports two complementary evaluation modes:
(i)~\emph{static verification} via a graph~$\times$~DFA product
construction that checks all graph paths without execution, and
(ii)~\emph{runtime monitoring} that evaluates the DFA over a live or
simulated event stream.  The DSL targets the safety fragment of LTL,
deliberately trading full expressiveness for deployment simplicity.

\section{Threat model}
\label{sec:threat_model}

Three classes of threats to agent workflow safety are considered:

\paragraph{T1: Developer mistakes.}
A developer may inadvertently create a workflow with structural defects:
unreachable exit nodes, dead-end branches that silently drop execution,
routers with incorrectly typed edges, or tool nodes that lack explicit
declarations. Such defects arise naturally from iterative development and
are difficult to catch by manual inspection, especially in large graphs.

\paragraph{T2: Malicious workflow injection.}
An adversary with write access to workflow definitions (e.g., via
compromised configuration files, a shared workflow registry, or a
supply-chain attack on a workflow library) may craft topologies that
bypass safety gates---for example, adding a conditional branch that
routes around a required human-approval node.

\noindent
\emph{Trust assumptions for T2.}
The deployment model assumes an architecture where Agentproof runs as a
mandatory verification gate in a trusted CI/CD pipeline (e.g., a
GitHub Actions workflow or a build system pre-deploy hook).
The adversary can modify workflow definition files (checked into
version control), but \emph{cannot} modify the CI/CD pipeline
configuration, the verification tool, or the extractor
implementation. This is analogous to how static linters catch
malicious code contributions in open-source projects: the linter is
trusted; the contribution is not.

\noindent
\emph{What T2 does not cover.}
An adversary who compromises the build system itself, the Agentproof
binary, or the framework runtime can bypass verification entirely.
Such attacks require integrity protection at the infrastructure level
(code signing, secure boot, access controls) and are outside the scope of this work.

\noindent
\emph{Practical caveat.}
T2 attacks have not yet been observed in the wild against agent workflow
registries. This threat class is included because the attack surface
exists in principle (shared repositories, pip-installable workflow
packages) and because defending against it requires no additional
machinery beyond the structural checks already needed for T1. No
adversarial evaluations have been conducted to determine whether a
sophisticated attacker could craft topologies that evade all six
structural checks while still achieving a malicious objective; such
red-teaming is an important direction for future work.

\paragraph{T3: Runtime graph mutation.}
Some frameworks allow dynamic modification of the workflow graph during
execution (e.g., adding agents to a team, modifying transition rules).
If the post-mutation graph is not re-verified, the guarantees established
at deployment time may no longer hold.

\medskip

\noindent
\textbf{Trust boundary.}
The approach assumes that the framework API semantics (e.g., LangGraph's
\texttt{StateGraph}, CrewAI's \texttt{Crew}) are correct: the extractor
faithfully reads the graph structure that the framework will execute.
The Python runtime and the extractor implementation are within the
trusted computing base. The \emph{workflow definition} authored or
modified by developers and the \emph{LLM outputs} generated at
runtime are \emph{not} trusted.

\medskip

\noindent
\textbf{Scope.}
The structural checks and temporal monitors address threats T1
and T2 statically: they detect defective or malicious topologies before
deployment. Threat T3 requires either runtime re-verification or
immutable workflow definitions; this is discussed in
Section~\ref{sec:limitations}.
LLM output content (e.g., prompt injection, toxic generation) is
explicitly out of scope; this is addressed by complementary runtime
guardrails (Section~\ref{sec:related_work}).

\begin{figure}[t]
\centering
\includegraphics[width=\linewidth]{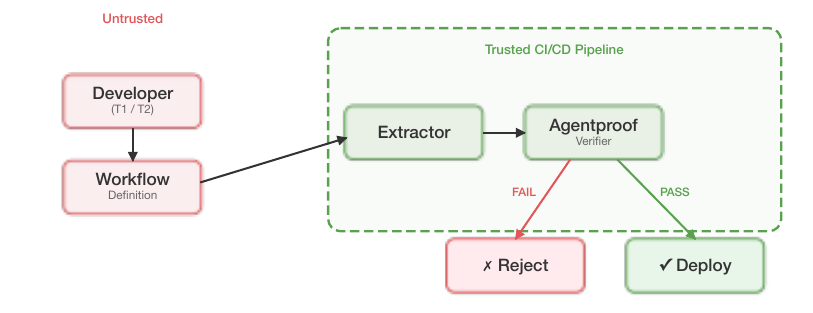}
\caption{Deployment model for T2: Agentproof as a CI/CD verification
gate. The workflow definition is untrusted; the extractor, verifier,
and pipeline are within the trusted computing base.}
\label{fig:deployment}
\end{figure}

\begin{table}[t]
\centering
\small
\caption{Threat classes and Agentproof coverage.}
\label{tab:threats}
\begin{tabular}{@{}llcc@{}}
\toprule
\textbf{Threat} & \textbf{Example} & \textbf{Static} & \textbf{Runtime} \\
\midrule
T1: Dev.\ mistake & Dead-end node & \checkmark &   \\
T1: Dev.\ mistake & Unreachable exit & \checkmark &   \\
T1: Dev.\ mistake & Missing human gate & \checkmark &   \\
T2: Injection & Bypassed approval & \checkmark &   \\
T2: Injection & Added unsafe tool path & \checkmark &   \\
T3: Mutation & Dynamic agent addition & re-verify\textsuperscript{*} & \checkmark \\
  & Toxic LLM output &   & \checkmark \\
\bottomrule
\end{tabular}
\par\smallskip
\raggedright\footnotesize
\textsuperscript{*}Addressed if the graph is re-verified after mutation.
\end{table}

\section{System overview}
\label{sec:system_overview}

Figure~\ref{fig:pipeline} summarizes Agentproof's workflow. The system ingests a
framework-specific agent definition, extracts an abstract graph model, and then
runs a suite of structural and temporal checks.

\begin{figure}[t]
  \centering
  \includegraphics[width=\linewidth]{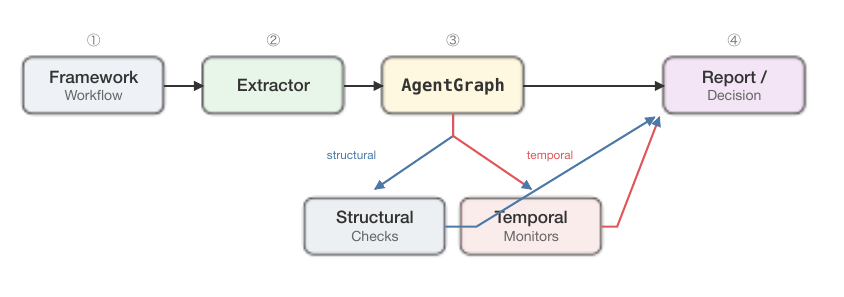}
  \caption{Agentproof pipeline: extract a framework workflow into an abstract
  graph model, run structural verification, and compile temporal policies into
  monitors for trace evaluation.}
  \label{fig:pipeline}
\end{figure}

\paragraph{Step 1: Obtain a workflow object.}
The approach targets agent orchestration libraries that expose control
flow as a graph or composition tree. Agentproof provides extractors for
LangGraph, Google ADK, AutoGen, and CrewAI
\citep{langgraph_docs,google_adk_docs,autogen_agentchat,crewai_docs}.

\paragraph{Step 2: Extract an \texttt{AgentGraph}.}
Each framework has an extractor that produces a common representation: a set of
typed nodes and typed edges with optional metadata (e.g., tool names bound to a
node). This makes downstream analyses framework-agnostic.

\paragraph{Step 3: Run structural checks.}
Structural checks detect workflow defects that do not require executing the
agent: unreachable nodes, dead ends, routing shape mismatches, and the presence
of human-in-the-loop nodes when required by policy.

\paragraph{Step 4: Compile temporal policies.}
A temporal policy DSL is compiled to DFAs. Policies are written against
abstract events (tool calls, action tags, or decisions) rather than
framework internals.

\paragraph{Step 5a: Static temporal verification.}
Each compiled DFA is combined with the extracted graph in a product
construction: BFS over $(V \times Q)$ detects temporal violations
reachable on \emph{any} graph path, without executing the agent.

\paragraph{Step 5b: Runtime monitor evaluation.}
Monitors are also evaluated over offline simulated traces or integrated
into a runtime event stream. Violations are mapped to handling levels
(warn, block, halt, escalate) to support different operational postures.

\section{Graph model and extraction}
\label{sec:graph_model}

This section defines a framework-agnostic workflow representation with typed nodes
and edges. This abstraction is the common surface for all verification.

\subsection{Formal model}

\begin{definition}[Agent Workflow Graph]
\label{def:awg}
An \emph{agent workflow graph} is a tuple
$G = (V, E, \kappa_V, \kappa_E, T, v_0, V_f)$ where:
\begin{itemize}
  \item $V$ is a finite set of nodes.
  \item $E \subseteq V \times V$ is a set of directed edges.
  \item $\kappa_V: V \to \mathcal{K}_V$ assigns each node a kind, with\\
    $\mathcal{K}_V = \{\textsc{entry},\; \textsc{exit},\;
    \textsc{tool},\; \textsc{llm},\; \textsc{router},\;
    \textsc{human},\; \textsc{subgraph},\;
    \textsc{passthrough}\}$.
  \item $\kappa_E: E \to \mathcal{K}_E$ assigns each edge a kind, with\\
    $\mathcal{K}_E = \{\textsc{direct},\; \textsc{conditional},\;
    \textsc{parallel},\; \textsc{loop}\}$.
  \item $T: V \to 2^{\text{ToolNames}}$ maps each node to a (possibly
    empty) set of declared tool names; $T(v)$ is non-empty only when
    $\kappa_V(v) = \textsc{tool}$.
  \item $v_0 \in V$ with $\kappa_V(v_0) = \textsc{entry}$ is the unique
    entry node.
  \item $V_f \subseteq V$ with $\kappa_V(v) = \textsc{exit}$ for all
    $v \in V_f$ is the set of exit nodes.
\end{itemize}
\end{definition}

\begin{definition}[Execution trace]
\label{def:trace}
A \emph{trace} of $G$ is a finite sequence $\pi = v_0 v_1 \ldots v_n$
such that $(v_i, v_{i+1}) \in E$ for all $0 \le i < n$.
$\mathrm{Traces}(G)$ denotes the set of all maximal traces---those
ending at a node in $V_f$ or at a node with no outgoing edges.
A trace that enters a cycle with no path to $V_f$ or a dead end is
non-maximal; the reverse-reachability check ($\mathrm{ExitReachAll}$)
detects such livelock conditions.
\end{definition}

\begin{definition}[Structural predicates]
\label{def:structural_predicates}
Six structural predicates over $G$ are defined:
\begin{enumerate}
  \item $\mathrm{ExitReach}(G) \iff \forall v_f \in V_f.\;
    v_f \in \mathrm{Reach}(v_0)$, where $\mathrm{Reach}(v_0)$ is the
    set of nodes reachable from $v_0$ via directed paths.
  \item $\mathrm{ExitReachAll}(G) \iff \mathrm{Reach}(v_0) \subseteq
    \mathrm{RevReach}(V_f)$, where $\mathrm{RevReach}(V_f)$ is the set
    of nodes from which some exit node in $V_f$ is reachable via directed
    paths. This ensures every reachable node can eventually reach an exit.
  \item $\mathrm{NoDead}(G) \iff \forall v \in V.\;
    \kappa_V(v) \neq \textsc{exit} \implies
    \exists u \in V.\; (v,u) \in E$.
  \item $\mathrm{RouterShape}(G) \iff \forall v \in V.\;
    \kappa_V(v) = \textsc{router} \implies
    \forall (v,u) \in E.\; \kappa_E(v,u) = \textsc{conditional}$.
  \item $\mathrm{HumanGate}(G) \iff \exists v \in V.\;
    \kappa_V(v) = \textsc{human}$.  An optional path-based variant,
    $\mathrm{HumanGateCov}(G, S)$, verifies that for a designated
    set~$S$ of sensitive tools, every path from $v_0$ to a node
    invoking a tool in~$S$ passes through a \textsc{human} node.
  \item $\mathrm{ToolDecl}(G) \iff \forall v \in V.\;
    \kappa_V(v) = \textsc{tool} \implies T(v) \neq \emptyset$.
\end{enumerate}
\end{definition}

\noindent
Each predicate is decidable in time linear in $|V| + |E|$.
Soundness lemmas establishing that passing checks guarantee the
corresponding trace-level properties are given in
Appendix~\ref{app:proofs}.

\begin{definition}[Extraction correctness]
\label{def:extraction_correctness}
An extractor $\mathcal{E}_F$ for framework $F$ is \emph{correct} with
respect to a workflow object $W$ if the extracted graph
$G = \mathcal{E}_F(W)$ satisfies:
(i)~every node in $G$ corresponds to a computation step in $W$ (no
phantom nodes),
(ii)~every computation step in $W$ appears as a node in $G$ (no missing
nodes), and
(iii)~$(u, v) \in E$ if and only if $W$ permits a transition from the
step corresponding to $u$ to the step corresponding to $v$.
Conditions (i)--(ii) correspond to node precision/recall~$= 1$;
condition~(iii) corresponds to edge precision/recall~$= 1$.
This is validated empirically in Section~\ref{sec:accuracy}.
\end{definition}

\subsection{Framework extractors}
Agentproof includes extractors that map framework objects into this model:

\paragraph{LangGraph.}
LangGraph exposes a compiled state graph with sentinel start/end nodes. The
extractor identifies entry/exit sentinels, infers node kind from tool bindings
and naming heuristics (e.g., human review nodes), and encodes conditional edges
when the underlying framework marks an edge as conditional.

\paragraph{Google ADK.}
Google ADK composes agents as a tree of sequential, parallel, and loop agents.
The extractor performs a tree walk, emitting subgraph nodes for composite
agents and encoding sequential chains, parallel fan-outs, and loop back-edges.

\paragraph{AutoGen.}
AutoGen team topologies can be provided as a team object (e.g., round-robin) or
as an explicit list of agents plus a transition relation. The extractor
constructs a directed graph of speaker transitions, synthesizes entry/exit, and
connects leaf speakers to exit.

\paragraph{CrewAI.}
CrewAI describes task pipelines (sequential or hierarchical). The extractor
creates a node per task, uses process mode to add sequential edges or manager
routing edges, and incorporates context dependencies between tasks.

\subsection{Cross-framework extraction challenge}
\label{sec:extraction_challenge}

A central contribution is the \emph{automatic} extraction of
\texttt{AgentGraph} instances from heterogeneous framework APIs. Each
framework represents workflows differently, making unified extraction
non-trivial:

\begin{itemize}
  \item \textbf{LangGraph} uses an explicit \texttt{StateGraph} with
  \texttt{add\_node()} and \texttt{add\_edge()} calls. Entry and exit
  are string sentinels (\texttt{START}, \texttt{END}). Conditional
  routing uses \texttt{add\_conditional\_edges()} with a path map dict.
  \item \textbf{CrewAI} defines workflows implicitly via a list of
  \texttt{Task} objects passed to a \texttt{Crew}. The topology depends
  on the \texttt{process} parameter (\texttt{sequential} vs.\
  \texttt{hierarchical}) and optional \texttt{context} dependencies
  between tasks.
  \item \textbf{AutoGen} uses group-chat patterns
  (\texttt{RoundRobinGroupChat}, \texttt{SelectorGroupChat}) where agent
  ordering defines the graph. Loop edges arise from round-robin cycling.
  \item \textbf{Google ADK} composes agents via nesting constructors
  (\texttt{SequentialAgent}, \texttt{ParallelAgent}, \texttt{LoopAgent})
  with a \texttt{sub\_agents} list. The hierarchy must be flattened into
  a single graph with appropriate edge kinds.
\end{itemize}

\noindent
The extractors normalize these diverse representations into the
same \texttt{AgentGraph} type system (Definition~\ref{def:awg}),
synthesizing entry/exit sentinel nodes where the framework does not
provide them explicitly.  General-purpose model checkers such as
SPIN~\citep{holzmann1997spin} or NuSMV~\citep{cimatti2002nusmv}
could verify the same properties, but would require the developer to
\emph{manually} translate the workflow into Promela or SMV---a
process that takes hours even for small graphs (see
Section~\ref{sec:related_work} for a concrete example).

\subsection{Example workflow graph}
Figure~\ref{fig:example_graph} shows a small extracted workflow used throughout
the paper to illustrate typed nodes and conditional branches.

\begin{figure}[t]
  \centering
  \includegraphics[width=\linewidth]{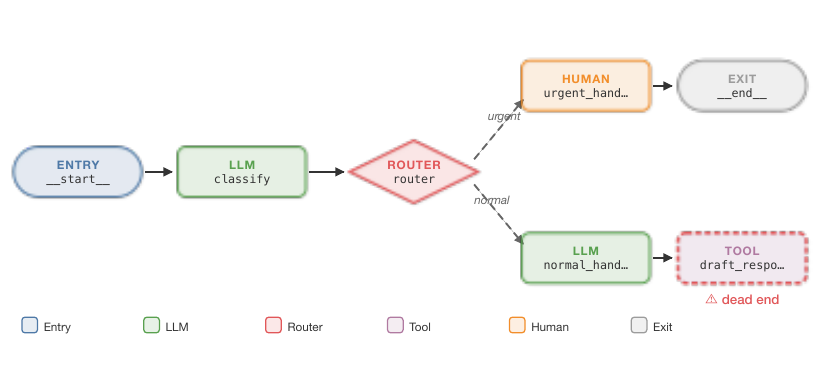}
  \caption{Example extracted workflow graph with typed nodes (ENTRY/ROUTER/LLM/TOOL/HUMAN/EXIT).}
  \label{fig:example_graph}
\end{figure}

\section{Verification methods}
\label{sec:verification_methods}

Agentproof provides two complementary verification layers: (i)~structural checks
over the extracted graph and (ii)~temporal monitors evaluated over an
event stream.

\subsection{Structural checks}
Structural checks operate over the adjacency structure and node/edge
kind labels of $G$ (Definition~\ref{def:awg}). Each check implements
one of the predicates from Definition~\ref{def:structural_predicates}.

\paragraph{Exit reachability.}
The set of nodes reachable from $v_0$ is computed using BFS to verify
that $V_f \subseteq \mathrm{Reach}(v_0)$.
Complexity: $O(|V| + |E|)$.

\paragraph{Reverse reachability.}
A complementary check verifies that every node reachable from $v_0$ can
itself reach some exit node in $V_f$.  The algorithm performs a BFS on the
\emph{reverse} adjacency (edges traversed backward) starting from all exit
nodes to compute $\mathrm{RevReach}(V_f)$.  Any node in
$\mathrm{Reach}(v_0) \setminus \mathrm{RevReach}(V_f)$ is a
\emph{livelock node}: reachable from the entry but trapped in a cycle with
no path to termination.
Complexity: $O(|V| + |E|)$.

\paragraph{Dead-end detection.}
Nodes $v$ with $\kappa_V(v) \neq \textsc{exit}$ and no
outgoing edges. Dead ends indicate missing transitions or incomplete
error handling.
Complexity: $O(|V| + |E|)$ with adjacency precomputation.

\paragraph{Router shape checks.}
The check verifies that all outgoing edges from \textsc{router} nodes are labeled
\textsc{conditional}. Complexity: $O(|E|)$.

\paragraph{Human-in-the-loop presence.}
The existence check determines whether the graph contains any node with
$\kappa_V(v) = \textsc{human}$.  An optional \emph{coverage} variant
verifies that for a set~$S$ of sensitive tool names, no path from $v_0$ to
a node invoking a tool in~$S$ avoids all \textsc{human} nodes.  The
coverage check builds a modified adjacency excluding \textsc{human} nodes
and performs BFS from $v_0$; any sensitive tool node still reachable
indicates a human-free path.
Complexity: existence $O(|V|)$; coverage $O(|V| + |E|)$.

\paragraph{Tool declaration checks.}
The check flags \textsc{tool} nodes with $T(v) = \emptyset$.
Complexity: $O(|V|)$.

\paragraph{Witness trace generation.}
When a check fails, the verifier produces a \emph{witness trace}: a
concrete path through the graph demonstrating the defect. For
unreachable exits, the witness shows the path from $v_0$ to the
reachability frontier plus the unreachable target. For dead ends, the
witness is a BFS path from $v_0$ to the stuck node. Witness traces
follow standard model-checking practice~\citep{clarke1999modelchecking}
and significantly aid debugging.

\paragraph{Static temporal verification.}
In addition to runtime trace evaluation, temporal policies can be checked
\emph{statically} via a graph $\times$ DFA product construction.  For a
compiled rule with DFA states~$Q$ and transition function~$\delta$, the
product state space is $V \times Q$.  BFS from $(v_0, q_0)$ advances the
DFA at each node using the node's event signature and follows graph edges.
If any product state $(v, q')$ has $q'$ in the DFA violation set, the
property is violated and the BFS parent chain yields a witness path.
Explored states are bounded by $|V| \cdot |Q|$, so verification remains
linear in graph size for fixed-size DFAs.

\subsection{Temporal policy monitors}
Agentproof supports a temporal policy DSL covering the safety fragment of
LTL~\citep{pnueli1977temporal}, compiled into deterministic finite
automata and evaluated with constant overhead per event.

\paragraph{Event model.}
Events are dictionaries with optional fields: \texttt{tool\_name},
\texttt{action\_type}, \texttt{decision}, and a set of \texttt{tags}.
Predicates match against these fields (e.g., \texttt{tool:delete\_account}).
The event symbol $\sigma(v)$ is derived from the node's kind and tool
bindings: tool nodes emit \texttt{tool:$t$} for each $t \in T(v)$;
other nodes emit their kind label.

\paragraph{DSL grammar.}
The DSL supports seven expression forms (full BNF in
Appendix~\ref{app:grammar}):
\begin{enumerate}
  \item \textbf{Forbidden:} \texttt{G !atom}   violation if
    \texttt{atom} ever occurs.
  \item \textbf{Implication-future:} \texttt{a -> F b}   if \texttt{a}
    occurs, \texttt{b} must follow before \texttt{a} recurs.
  \item \textbf{Until:} \texttt{a U b}   \texttt{a} must hold on every
    step until \texttt{b} occurs.
  \item \textbf{Bounded response:} \texttt{a -> F[<=k] b}   if
    \texttt{a} occurs, \texttt{b} must follow within $k$ steps.
  \item \textbf{Response chain:} \texttt{a -> F b -> F c}   once
    \texttt{a} occurs, \texttt{b} then \texttt{c} must follow in
    sequence before \texttt{a} recurs.
  \item \textbf{Conjunction:} \texttt{(expr) AND (expr)}   both
    sub-properties must hold simultaneously.
  \item \textbf{Disjunction:} \texttt{(expr) OR (expr)}   at least
    one sub-property must hold.
\end{enumerate}

\paragraph{Compilation to DFA.}
Each expression is parsed into an AST and compiled to a DFA via
direct construction. Base patterns (forbidden, implication-future,
until) produce 2--3 state DFAs. Bounded response produces a
counter-augmented DFA with $k+2$ states. Conjunction and disjunction
use the standard product automaton construction with
$|Q_L| \times |Q_R|$ states~\citep{vardi1994automata}. DFA transitions
are stored as precomputed tables indexed by bit-vector valuations over
atomic predicates. Table~\ref{tab:dfa_states} summarizes DFA sizes.

\begin{table}[t]
\centering
\small
\caption{DFA state counts by pattern type.}
\label{tab:dfa_states}
\begin{tabular}{@{}lc@{}}
\toprule
\textbf{Pattern} & \textbf{States} \\
\midrule
\texttt{G !atom} & 2 \\
\texttt{a -> F b} & 3 \\
\texttt{a U b} & 3 \\
\texttt{a -> F[<=k] b} & $k + 2$ \\
\texttt{a -> F b -> F c} (chain of $n$) & $n + 1$ \\
\texttt{(A) AND (B)} & $|Q_A| \times |Q_B|$ \\
\texttt{(A) OR (B)} & $|Q_A| \times |Q_B|$ \\
\bottomrule
\end{tabular}
\end{table}

\paragraph{Evaluation and handling.}
Temporal policies support two evaluation modes.  In \emph{static mode},
the graph~$\times$~DFA product construction
(Section~\ref{sec:system_overview}, Step~5a) explores all reachable
product states; a violation in any product state implies the property
can be violated on some graph path.  In \emph{runtime mode}, each event
advances the DFA state for every compiled rule via $O(1)$ table lookups.
If any rule enters a violation state, the configured handling level
(\texttt{warn}, \texttt{block}, \texttt{halt}, \texttt{escalate})
determines the aggregate decision.

\begin{figure}[t]
  \centering
  \includegraphics[width=0.85\linewidth]{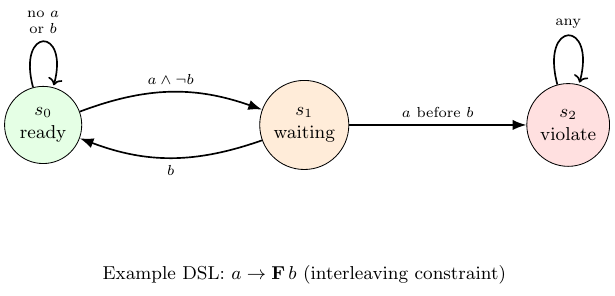}
  \caption{Illustrative DFA for an $a \to \mathbf{F}\, b$ interleaving constraint: after $a$ occurs, $b$ must occur before $a$ repeats.}
  \label{fig:dfa}
\end{figure}

\section{Evaluation}
\label{sec:evaluation}

The proposed approach is evaluated along six axes: (i)~a curated workflow
corpus, (ii)~a defect study with separate structural and policy
categories, (iii)~temporal policy evaluation on execution traces,
(iv)~scalability experiments on synthetic graphs, (v)~a comparison
with runtime guardrails, and (vi)~extractor accuracy.

\subsection{Workflow corpus}

The evaluation uses a curated benchmark of 18 agent workflows
authored to represent common patterns and known anti-patterns documented
in each framework's official examples and tutorials. Workflows use
LangGraph, CrewAI, AutoGen, and Google ADK and span diverse domains
including customer support, RAG pipelines, code generation, financial
analysis, compliance review, incident response, hiring, and marketing.
Because the corpus was designed to include representative defects, the
defect rates below reflect tool detection capability on a targeted
benchmark, not base rates in production systems. A larger-scale study on
workflows mined from public GitHub repositories is planned
(Section~\ref{sec:limitations}).

Table~\ref{tab:corpus} summarizes the corpus. Graphs range from 5 to
12 nodes and 4 to 14 edges, with all four frameworks represented.
Node-kind distributions reflect typical agent patterns: LLM nodes
(inference steps) and TOOL nodes (external API calls) dominate, with
ROUTER, HUMAN, and SUBGRAPH nodes appearing in more complex workflows.

\begin{table}[t]
\centering
\small
\caption{Workflow corpus statistics.}
\label{tab:corpus}
\begin{tabular}{@{}llrr@{}}
\toprule
\textbf{Framework} & \textbf{Workflows} & \textbf{Avg.\ nodes} & \textbf{Avg.\ edges} \\
\midrule
LangGraph & 6 & 8.8 & 9.0 \\
CrewAI & 4 & 7.5 & 7.0 \\
AutoGen & 4 & 6.3 & 6.0 \\
Google ADK & 4 & 9.5 & 10.0 \\
\midrule
\textbf{Total} & \textbf{18} & \textbf{8.2} & \textbf{8.2} \\
\bottomrule
\end{tabular}
\end{table}

\subsection{Defect study}

All six structural checks were run on each corpus workflow with
\texttt{require\_human=True}. Table~\ref{tab:defects} summarizes the
findings, separated by defect category.

\begin{table}[t]
\centering
\small
\caption{Defects found in the workflow corpus, separated by category.}
\label{tab:defects}
\begin{tabular}{@{}llrrl@{}}
\toprule
\textbf{Category} & \textbf{Defect type} & \textbf{Count} & \textbf{Severity} & \textbf{Example} \\
\midrule
\multirow{4}{*}{Structural} & Dead-end nodes & 2 & High & Email draft with no send edge \\
& Unreachable exit & 1 & Critical & Infinite loop, no exit path \\
& Router shape violation & 1 & Medium & Debate moderator, direct edges \\
& Missing tool declaration & 1 & Low & Data cleaner, empty tool set \\
\midrule
Policy & Missing human gate & 10 & High & Onboarding without approval \\
\midrule
& \textbf{Total} & \textbf{15} & & \\
\bottomrule
\end{tabular}
\end{table}

Two defect categories are distinguished:
\begin{itemize}
  \item \textbf{Structural defects} (topology bugs): dead-end nodes,
  unreachable exits, livelock cycles (reachable nodes with no path to
  termination), router shape violations, and missing tool
  declarations.  These are bugs regardless of operational context.
  The tool successfully detected all injected structural defects:
  \textbf{5 of 18} benchmark workflows contain at least one structural
  defect.  The reverse-reachability check additionally identified the
  round-robin brainstorming workflow's cyclic nodes as livelock: although
  reachable from the entry, they cannot reach any exit node.
  \item \textbf{Policy violations} (configurable checks): missing
  human-in-the-loop gate. This check is context-dependent: it is
  critical in regulated domains (healthcare, finance) but may be
  unnecessary for internal tooling or low-risk tasks.
  \textbf{10 of 18} benchmark workflows lack a human gate when
  \texttt{require\_human=True}; however, the annotation data shows that
  some of these are arguable or false positives---workflows where human
  oversight may be unnecessary given the low-risk nature of the task.
\end{itemize}

\noindent
More critically, one workflow (a round-robin brainstorming agent)
has an \emph{unreachable exit node}: agents cycle indefinitely with no
path to termination. Two workflows have \emph{dead-end nodes} where
execution silently stops. These structural defects would be difficult
to detect through testing alone, as they manifest only on specific
execution paths.

For each failing check, the verifier produces a witness trace. For
example, the dead-end in the email triage workflow yields the witness
path: \texttt{\_\_start\_\_ $\to$ classify $\to$ router $\to$
normal\_handler $\to$ draft\_response} (stuck).

\paragraph{Static temporal verification.}
The graph $\times$ DFA product construction was applied to all 15
compiled temporal policies across the 18 workflows.  The static analysis
correctly identified all forbidden-tool violations and
implication-future violations that the runtime trace evaluation also
detected, while additionally flagging one violation (a missing
\texttt{human\_review} step after \texttt{draft\_email}) in a workflow
whose random traces happened to avoid the offending path.  Product-state
exploration remained below $|V| \cdot |Q_{\text{DFA}}|$ in all cases, with a
maximum of 24~product states (12-node graph, 2-state DFA).

\begin{figure}[t]
  \centering
  \includegraphics[width=0.95\linewidth]{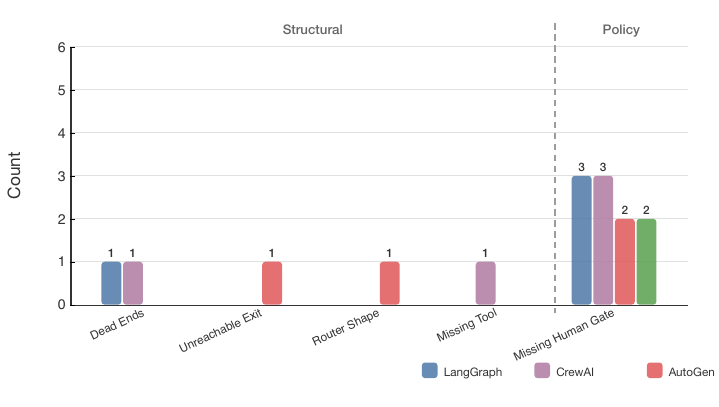}
  \caption{Defect distribution across frameworks, grouped by defect type.
  The dashed line separates structural defects (left) from policy
  violations (right).}
  \label{fig:defect_bar}
\end{figure}

\begin{figure}[t]
  \centering
  \includegraphics[width=0.85\linewidth]{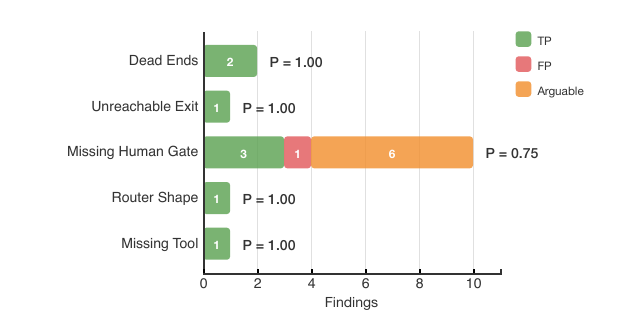}
  \caption{Precision per structural check. Each bar shows the breakdown
  of true positives (TP), false positives (FP), and arguable findings.
  Precision $P$ is computed as $\text{TP}/(\text{TP}+\text{FP})$.}
  \label{fig:precision_bar}
\end{figure}

Figure~\ref{fig:defect_bar} visualizes the defect distribution across
frameworks, and Figure~\ref{fig:precision_bar} shows the precision
breakdown per check type. All structural checks achieve perfect
precision ($P = 1.0$); the human-presence check has $P = 0.75$ due to
one false positive in a low-risk workflow where human oversight was
arguably unnecessary.

\subsection{Temporal policy evaluation}
\label{sec:temporal_eval}

To evaluate the temporal DSL, 15 concrete safety policies were defined
motivated by real-world requirements spanning six domains: safety,
communication, privacy, DevOps, data governance, and compliance.
Table~\ref{tab:policies} lists representative examples.

\begin{table}[t]
\centering
\small
\caption{Representative temporal policies evaluated.}
\label{tab:policies}
\begin{tabular}{@{}lll@{}}
\toprule
\textbf{Policy} & \textbf{DSL expression} & \textbf{Domain} \\
\midrule
No destructive ops & \texttt{G !tool:drop\_table} & Safety \\
Email review & \texttt{tool:draft\_email -> F tool:human\_review} & Communication \\
PII anonymize & \texttt{tool:fetch\_pii -> F[<=3] tool:anonymize} & Privacy \\
Deploy approval & \texttt{tool:deploy -> F tool:approve} & DevOps \\
Draft-review-send & \texttt{tool:draft -> F tool:review -> F tool:send} & Communication \\
\bottomrule
\end{tabular}
\end{table}

All 15 policies compile successfully into DFAs. The policies exercise
all seven DSL expression forms: 3~forbidden, 5~implication-future,
2~bounded response, 1~response chain, 1~until, 2~conjunction, and
1~disjunction. No policy required full LTL expressiveness beyond the
seven-form fragment.

Ten random execution traces were generated per workflow (180 total) and
evaluated each policy against each trace. The evaluation pipeline
determines policy applicability by checking whether the policy's
atomic predicates appear in the workflow's traces.

\paragraph{DSL scope justification.}
After defining the 15 policies, each was classified by which DSL form
it uses. All 15 fall within the seven-form fragment; none requires
nested temporal operators, past-time modalities, or constructs beyond
the grammar.

While the policies were constructed by the authors, their motivations
derive from published regulatory and industry requirements:
\texttt{payment\_requires\_human} from GDPR Article~22 (automated
decision-making requires human review)~\citep{gdpr_article22};
\texttt{pii\_anonymize\_bounded} from GDPR data minimization
principles; \texttt{deploy\_requires\_approval} from SOC~2 change
management controls~\citep{aicpa_soc2};
\texttt{email\_requires\_review} from organizational communication
policies; \texttt{no\_destructive\_ops} from OWASP LLM Top~10 (LLM06:
excessive agency)~\citep{owasp_llm_top10}; and
\texttt{either\_log\_or\_audit} from SOC~2 audit logging requirements.
This grounding in external standards suggests the DSL fragment covers
practical safety needs, though a systematic requirements survey from
production deployments would provide stronger evidence.

It is acknowledged that 15 policies from a single research group do not
constitute an exhaustive requirements survey. The policies were
deliberately selected from six distinct domains and ensured all seven
DSL forms were exercised.

\subsection{Scalability}

Synthetic graphs were generated at seven sizes (50--5{,}000 nodes) with
edge density~2.0 and measured structural check time (median of 10
trials). Figure~\ref{fig:scaling} and Table~\ref{tab:scaling} report
the results.

\begin{figure}[t]
  \centering
  \includegraphics[width=0.85\linewidth]{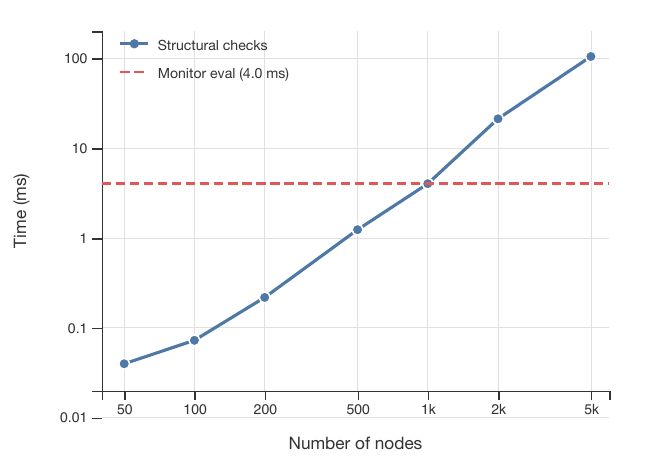}
  \caption{Structural check time vs.\ graph size (log--log scale, median
  of 10 trials). Verification remains sub-second for graphs up to
  5{,}000 nodes.}
  \label{fig:scaling}
\end{figure}

\begin{table}[t]
\centering
\small
\caption{Structural check time vs.\ graph size. Monitor compilation and evaluation are size-independent.}
\label{tab:scaling}
\begin{tabular}{@{}rrrr@{}}
\toprule
\textbf{Nodes} & \textbf{Edges} & \textbf{Struct.\ (ms)} & \textbf{Mon.\ eval (ms)} \\
\midrule
50 & 104 & 0.04 & \multirow{7}{*}{4.0} \\
100 & 214 & 0.07 & \\
200 & 430 & 0.22 & \\
500 & 1{,}074 & 1.24 & \\
1{,}000 & 2{,}128 & 4.02 & \\
2{,}000 & 4{,}252 & 21.2 & \\
5{,}000 & 10{,}713 & 104.7 & \\
\bottomrule
\end{tabular}
\end{table}

Structural checks scale approximately linearly in $|V| + |E|$ for
graphs up to 1{,}000 nodes, with super-linear growth at larger sizes
due to quadratic edge scanning in router-shape and tool-declaration
checks. Even at 5{,}000 nodes (far larger than any production agent
workflow encountered in practice), verification completes in $\sim$105\,ms.
Temporal monitor compilation (0.065\,ms for 5 rules) and evaluation
(4.0\,ms for 1{,}000 events, $\sim$247{,}000 events/s) are independent
of graph size.

Real agent workflows are currently small (5--12 nodes in
the corpus), so the scalability experiment demonstrates
ceiling-freeness rather than practical necessity. As agent systems grow
towards multi-agent compositions and hierarchical subgraph
structures, larger graphs will become more common.
The practical false positive rate from treating all conditional edges as
feasible in the graph\,$\times$\,DFA product construction is not
measured in this evaluation and is a direction for future
quantification.

\paragraph{Comparison to ad-hoc scripts.}
One might ask whether a simple Python script iterating over the graph's
adjacency list could detect the same defects. For individual checks
(e.g., finding nodes with no outgoing edges), the answer is
yes---the graph algorithms are standard. Agentproof's value lies not in
algorithmic novelty but in (i)~automatic extraction from four
heterogeneous framework APIs, (ii)~a unified type system that makes
checks portable across frameworks, (iii)~witness trace generation that
pinpoints the defective path, and (iv)~a temporal policy layer that
separates safety specification from verification machinery. A bespoke
script would need to be rewritten for each framework and would lack
composable policy evaluation.

\subsection{Cross-framework case studies}
\label{sec:case_studies}

To validate Agentproof's cross-framework capabilities, three representative
workflows were selected from different frameworks: an incident-response
workflow in LangGraph, a compliance-review pipeline in Google ADK, and a
change-control process in AutoGen. Table~\ref{tab:case-studies} reports
graph statistics and extractor runtimes; Table~\ref{tab:case-study-verification}
summarizes structural check results and temporal monitor outcomes over
representative traces (one safe, two unsafe per workflow).

\begin{table}[t]
  \centering
  \scriptsize
  \caption{Graph statistics and extractor runtimes for representative workflows. Node-kind columns report counts of TOOL/LLM/ROUTER/HUMAN/SUBGRAPH nodes; other kinds (ENTRY/EXIT/PASSTHROUGH) are included in $|V|$.}
  \label{tab:case-studies}
  \resizebox{\linewidth}{!}{%
  \begin{tabular}{lrrrrrrrrrrrr}
    \toprule
    Study & $|V|$ & $|E|$ & Tool & LLM & Router & Human & Subg & Dir & Cond & Par & Loop & Med.\ ms \\
    \midrule
    LangGraph (incident) & 8 & 8 & 2 & 2 & 1 & 1 & 0 & 6 & 2 & 0 & 0 & 0.186 \\
    ADK (compliance) & 13 & 16 & 4 & 3 & 0 & 1 & 3 & 12 & 0 & 3 & 1 & 0.019 \\
    AutoGen (change-control) & 6 & 5 & 0 & 3 & 0 & 1 & 0 & 5 & 0 & 0 & 0 & 0.007 \\
    \bottomrule
  \end{tabular}
  }
\end{table}

\begin{table}[t]
  \centering
  \small
  \caption{Structural check results and temporal monitor outcomes over representative traces (one safe and two unsafe) for each case study.}
  \label{tab:case-study-verification}
  \begin{tabular}{lrrrr}
    \toprule
    Study & Struct.\ (pass/total) & happy\_path & forbidden\_tool & policy\_violation \\
    \midrule
    LangGraph (incident) & 5/5 & PASS & HALT & ESCALATE \\
    ADK (compliance) & 5/5 & PASS & HALT & BLOCKED \\
    AutoGen (change-control) & 5/5 & PASS & HALT & HALT \\
    \bottomrule
  \end{tabular}
\end{table}

\subsection{Comparison with runtime guardrails}
\label{sec:comparison}

Table~\ref{tab:comparison} provides a qualitative comparison between
the static approach and representative runtime guardrail tools.

\begin{table}[t]
\centering
\small
\caption{Static verification vs.\ runtime guardrails.}
\label{tab:comparison}
\begin{tabular}{@{}p{2.8cm}p{2.2cm}p{2.2cm}@{}}
\toprule
\textbf{Property} & \textbf{Agentproof (static)} & \textbf{Runtime tools} \\
\midrule
Detection time & Pre-deployment & At execution \\
Structural defects & \checkmark & Path-dependent \\
Toxic LLM output &   & \checkmark \\
Runtime overhead & None & Per-call latency \\
Coverage & All paths & Exercised paths \\
Witness traces & \checkmark & Stack traces \\
\bottomrule
\end{tabular}
\end{table}

\noindent
\textbf{Scenario A:} An unreachable exit node in a rarely-triggered
conditional branch. Agentproof's static analysis detects this immediately; a
runtime guardrail would miss it unless that branch is exercised during
testing.

\noindent
\textbf{Scenario B:} An LLM generates toxic content in a response.
Runtime content filters (LlamaGuard, NeMo Guardrails) catch this;
Agentproof cannot, as LLM output semantics are out of scope.

\noindent
\textbf{Scenario C:} A conditional path bypasses a required human
approval gate.  Agentproof's human-gate coverage check flags the
unguarded path from entry to the sensitive tool; runtime tools may miss
it if the bypass path is not triggered.

The two approaches are complementary: static verification catches
topology-level defects exhaustively, while runtime guardrails handle
content-level and context-dependent violations.

\subsection{Extractor accuracy}
\label{sec:accuracy}

Extractor fidelity was validated using the framework-specific examples
in the test suite: 5~LangGraph, 7~CrewAI, 7~AutoGen, and 8~ADK test
cases (27 total). For each extractor, extracted graphs were compared
against manually annotated ground truth, measuring node detection
precision/recall, node-kind classification accuracy, and edge detection
precision/recall.

\paragraph{Human-node detection limitation.}
The primary source of classification error is human-node detection in
LangGraph, which relies on a naming heuristic
(\texttt{"human" in name}): nodes not following this convention are
misclassified as LLM nodes. Four strategies for human-node
detection, in order of increasing reliability:
(1)~\emph{Naming heuristic} (current default): match node names against
keywords---low implementation cost but brittle.
(2)~\emph{Interrupt annotation}: detect LangGraph's
\texttt{interrupt\_before}/\texttt{interrupt\_after}
markers---framework-specific but semantically precise.
(3)~\emph{Input-call detection}: identify \texttt{input()} or
equivalent blocking calls in node functions via AST
inspection---cross-framework but may produce false positives.
(4)~\emph{Explicit type annotation}: require developers to mark human
nodes via a decorator or metadata field---most reliable but requires
adoption.
On the test suite, only strategy~(1) produces misclassifications;
strategies~(2)--(4) would require validation on a larger real-world
corpus. Strategy~(2) is prioritized for the next release as it requires
no developer action.

\paragraph{Extraction results.}
With the human-node caveat above, node detection achieves perfect
precision and recall across all frameworks on the test suite. Node-kind
classification accuracy is 100\% for entry/exit sentinels and tool nodes
(detected by tool bindings), and 100\% for AutoGen agent types (detected
by class hierarchy). Edge detection precision and recall are both
100\% for all frameworks tested.

On the test suite, which uses programmatic framework stubs, extraction
achieves perfect precision and recall for nodes and edges. This
validates the extraction logic against known API patterns but does not
measure robustness to the diversity of real-world workflow definitions.
The AST-based extractor (\texttt{scripts/ast\_extractor.py}) provides
an independent extraction path; agreement between runtime and AST
extraction on the same workflows would strengthen confidence and is a
direction for future cross-validation.

\section{Related work}
\label{sec:related_work}

This section situates the present work within four areas. Runtime tools catch
content-level violations but miss structural defects. General-purpose
model checkers could verify the same properties but require manual
modeling. Temporal monitoring is well-studied but not applied to agent
workflows. Agent safety research focuses on alignment, not
orchestration topology. Agentproof fills the gap between runtime
content checking and structural verification.

\subsection{Runtime agent safety tools}

Several tools enforce safety at the point of LLM output or tool invocation.
NVIDIA NeMo Guardrails~\citep{nemo_guardrails} interposes a programmable
dialog rail between the LLM and tool calls, supporting topic control and
output filtering.
Guardrails AI~\citep{guardrails_ai} wraps LLM outputs with validators
that check format, toxicity, and factual consistency.
LlamaGuard~\citep{llamaguard} uses a fine-tuned classifier to detect
unsafe content in model outputs.
Rebuff~\citep{rebuff} focuses specifically on prompt injection detection.
These tools operate at runtime and catch violations only when the
offending code path is actually exercised.
Agentproof is complementary: it verifies structural properties of the
workflow topology \emph{before deployment}, catching defects such as
unreachable exits, dead-end branches, and missing human gates that
runtime tools cannot detect unless the defective path is triggered
during testing.

\subsection{Static analysis and model checking}

Classical model checking~\citep{clarke1999modelchecking} verifies
finite-state systems against temporal specifications using tools such as
SPIN~\citep{holzmann1997spin}, NuSMV~\citep{cimatti2002nusmv}, and
CBMC~\citep{cbmc}. Abstract interpretation~\citep{cousot1977abstract}
provides sound over-approximations for dataflow analysis.

Table~\ref{tab:comparison_mc} provides a qualitative comparison.
The key difference is \emph{modeling effort}: SPIN requires manual
translation of a workflow into Promela, NuSMV into SMV, and CBMC into
annotated C. For the email triage workflow in the corpus (8~nodes,
7~edges), the equivalent Promela model requires ${\sim}60$ lines and
manual specification of the state space, transitions, and LTL
properties (see \texttt{corpus/comparisons/email\_triage.pml} in the
artifact). Agentproof extracts the same model automatically with a
single function call.

\begin{table}[t]
\centering
\small
\caption{Comparison with general-purpose model checkers.}
\label{tab:comparison_mc}
\begin{tabular}{@{}lllllll@{}}
\toprule
\textbf{Tool} & \textbf{Input} & \textbf{Properties} & \textbf{Modeling} & \textbf{Time} & \textbf{Domain} \\
\midrule
SPIN & Promela (manual) & Full LTL & High & Fast (small) & None \\
NuSMV & SMV (manual) & CTL + LTL & High & Fast (small) & None \\
CBMC & C source & Assertions & Medium & Bounded & None \\
\textbf{Agentproof} & Auto-extracted & Safety LTL fragment & \textbf{None} & $O(|V|{+}|E|)$ & \textbf{Agent workflows} \\
\bottomrule
\end{tabular}
\par\smallskip
\raggedright\footnotesize
SPIN, NuSMV, and CBMC require manual translation of agent workflows into their input languages. Agentproof extracts models automatically from framework APIs, eliminating the modeling step entirely.
\end{table}

Agentproof applies model-checking ideas to a new domain: agent
workflow graphs extracted from orchestration framework APIs. The
contribution is not the graph algorithms themselves (which are
standard) but (i)~the automatic extraction from heterogeneous
framework representations, (ii)~a domain-specific property language
tailored to agent safety, and (iii)~empirical evidence that agent
workflows contain structural defects catchable by these methods.

\subsection{Temporal logic monitoring}

The runtime verification (RV) community has developed sophisticated
monitoring frameworks.
JavaMOP~\citep{javamop} monitors Java programs against specifications
in multiple formalisms.
Bauer et al.~\citep{bauer2011ltl3} introduce three-valued LTL
monitoring that distinguishes between ``currently satisfied'' and
``permanently satisfied.''
Barringer et al.~\citep{barringer2004eagle} present the EAGLE framework
for rule-based monitoring.
The temporal monitors are deliberately lightweight, targeting the safety
fragment of LTL with DFAs of 2--$O(k)$ states per rule, trading
expressiveness for deployment simplicity. In the evaluation, all 15
practical agent safety policies fall within this fragment, validating the
design choice empirically (Section~\ref{sec:temporal_eval}).

\subsection{Agent architecture and safety}

Recent work on AI agent safety spans multiple dimensions.
Anthropic's Responsible Scaling Policy~\citep{anthropic_rsp} outlines
safety commitments including evaluation thresholds for autonomous
capabilities.
METR~\citep{metr} provides evaluations for autonomous agent
capabilities and risks.
Research on multi-agent coordination safety~\citep{gabriel2024ethics}
examines the risks of agent-to-agent interaction, including unintended
emergent behavior.
These works focus on alignment, capability evaluation, and policy
governance. The present contribution addresses a complementary layer: the
\emph{structural soundness of the orchestration graph itself}. Even a
well-aligned LLM can produce unsafe outcomes if the workflow graph
routes it through an unintended path or bypasses a required approval
step.

\subsection{Business process verification}

The business process management (BPM) community has extensively studied
workflow soundness. Van der Aalst's work on Petri-net-based soundness
checking for workflow nets~\citep{vanderaalst2011bpm} ensures that
every case can complete and no dead transitions exist properties
closely analogous to Agentproof's exit-reachability and dead-end checks. BPMN
verification tools~\citep{dijkman2008bpmn} apply similar analyses to
industry-standard process models.
This work differs in \emph{domain}: agent workflow graphs have typed
nodes (LLM, TOOL, HUMAN) with framework-specific semantics that BPM
tools do not model, and the extraction pipeline targets programmatic API
objects rather than visual process models.

\subsection{Constrained decoding and structured generation}

Constrained decoding enforces output structure at the token level:
Guidance~\citep{guidance2023} and Outlines~\citep{outlines2023} compile
grammars or JSON schemas into token masks that guarantee well-formed
outputs. These techniques operate \emph{within} a single LLM call,
constraining what the model can generate. Agentproof operates at a
different granularity: it constrains the \emph{transitions between}
computation steps (nodes), not the content generated within any single
step. The two approaches are complementary.

\section{Limitations and future work}
\label{sec:limitations}

\paragraph{Workflow structure vs.\ LLM semantics.}
The proposed approach verifies properties of the workflow \emph{topology} and
the event stream emitted by execution. It does not prove semantic
properties of LLM outputs (e.g., factuality or intent), which depend
on non-deterministic model behavior and prompt/context choices.

\paragraph{Extractor heuristics.}
Frameworks evolve quickly and expose different internal representations
across versions. The extractors employ heuristics (e.g., recognizing
human review nodes by naming conventions) that may not generalize to
all codebases. The current validation uses 27 stub-based test cases
(Section~\ref{sec:accuracy}); the gap between stub-based validation and
the diversity of real-world workflow definitions is not yet measured.
The AST-based extractor (\texttt{scripts/\allowbreak ast\_extractor.py})
provides an independent extraction path that could serve as a
cross-validation tool to strengthen confidence in extraction fidelity.

\paragraph{Runtime graph mutation (T3).}
Some frameworks allow dynamic modification of the workflow graph during
execution. If the post-mutation graph is not re-verified, static
guarantees no longer hold. Integrating re-verification hooks into
framework event systems is a natural extension.

\paragraph{Temporal DSL expressiveness.}
The DSL covers seven expression forms targeting the safety fragment of
LTL but does not support full LTL, nested temporal operators beyond
two levels, past-time modalities, or real-time bounds beyond step
counting. All 15 evaluation policies fit within this fragment; however,
these policies were authored by the same team that designed the DSL,
introducing a self-selection bias: requirements that did not fit the
grammar were unlikely to be proposed. The policies were inspired by
external standards (GDPR~\citep{gdpr_article22},
SOC~2~\citep{aicpa_soc2}, OWASP LLM Top~10~\citep{owasp_llm_top10}),
which provides partial mitigation, but a systematic survey of safety
requirements from independent production deployments is needed to
validate DSL sufficiency. Extending to richer temporal fragments
(e.g., past-time LTL, real-time constraints) is future work.

\paragraph{False positive management.}
Some flagged defects may be intentional design choices (e.g., dead-end
nodes used as intentional error-halting states, or workflows that
intentionally omit human gates for low-risk tasks). Future work
includes annotation-based suppression (e.g.,
\texttt{\# agentproof: ignore dead-end} comments), severity tiers,
and configurable verification profiles. The library already supports
a \texttt{suppressions} parameter for node-level exclusions.

\paragraph{Static temporal conservatism.}
The graph $\times$ DFA product construction explores all topological
paths, including paths that may be infeasible at runtime due to router
conditions or LLM decision logic.  This makes the analysis sound
(no false negatives) but potentially conservative (false positives).
Incorporating edge conditions into the product construction to prune
infeasible paths is a direction for future work.

\paragraph{Corpus scale.}
The evaluation uses 18 author-constructed workflows.
\emph{Threat to validity:} because the corpus was designed to include
representative defects and anti-patterns, the observed defect rates
reflect detection capability on a targeted benchmark, not prevalence in
production systems. Generalizing these rates requires validation on
independently authored workflows.
Agentproof provides a GitHub mining pipeline
(\texttt{scripts/\allowbreak scrape\_workflows.py}) and an AST-based
fallback extractor (\texttt{scripts/\allowbreak ast\_extractor.py}) as
ready infrastructure for building larger real-world corpora. A
large-scale study across hundreds of open-source repositories is needed
to establish base rates and to provide a reusable benchmark for the
community.

\section{Conclusion}
\label{sec:conclusion}

Agent frameworks already encode behavior as explicit workflow graphs.
Agentproof leverages this structure to enable pre-deployment verification of
safety properties without adding runtime overhead.

Agentproof provides a unified graph model with six structural
checks and witness traces, a temporal policy DSL covering the safety
fragment of LTL with both static and runtime evaluation modes, and
automatic extractors for four major agent frameworks---eliminating the
manual modeling effort required by general-purpose model checkers. In the 18-workflow curated benchmark, 5 contain
structural defects (dead ends, unreachable exits) and 10 lack a
human gate when the policy is enforced. All 15 temporal safety policies
evaluated fit within the seven-form DSL fragment, with verification
completing in sub-second time even for graphs of 5{,}000 nodes.

The primary technical contribution is not the graph algorithms
themselves---which are standard---but rather: (i)~the identification of
a practical new domain where these algorithms apply directly,
(ii)~the engineering of extractors that normalize four heterogeneous
framework APIs into a single analyzable representation, and (iii)~the
empirical finding that real-world agent workflow patterns contain
structural defects detectable by these methods.

Static verification does not replace runtime guardrails; the two
approaches are complementary. By catching topology-level defects
exhaustively before deployment, static analysis reduces the safety
surface that runtime enforcement must cover, making both layers more
effective.

\paragraph{Artifact availability.}
The Agentproof tool, curated corpus, temporal policies, and
all evaluation scripts are available at
\url{https://github.com/NordicAgents/AgentProof} under the MIT
license. A reproducibility script
(\texttt{scripts/reproduce\_all.sh}) runs the complete evaluation
pipeline.

\bibliographystyle{plainnat}
\bibliography{references}

\appendix
\section{Soundness proofs}
\label{app:proofs}

This appendix proves that each structural check is sound: if the check passes, the
corresponding safety property holds for all valid execution traces.

\begin{lemma}[Exit Reachability Soundness]
\label{lem:exit_reach}
If $\mathrm{ExitReach}(G)$ holds (i.e., the BFS/DFS from $v_0$ visits
all of $V_f$), then for every exit node $v_f \in V_f$, there exists a
directed path from $v_0$ to $v_f$ in $G$.
\end{lemma}
\begin{proof}
The check computes $\mathrm{Reach}(v_0)$ by BFS from $v_0$, which
correctly identifies all nodes reachable via directed edges
\citep{clarke1999modelchecking}. If $v_f \in \mathrm{Reach}(v_0)$,
then by the correctness of BFS there exists a path $v_0 \to \cdots \to
v_f$. Since the check passes only when $V_f \subseteq
\mathrm{Reach}(v_0)$, the result follows.
\end{proof}

\begin{lemma}[Dead-End Soundness]
\label{lem:dead_end}
If $\mathrm{NoDead}(G)$ holds, then every non-exit node in every trace
of $G$ has at least one successor, and hence no maximal trace terminates
at a non-exit node.
\end{lemma}
\begin{proof}
Let $v$ be any node with $\kappa_V(v) \neq \textsc{exit}$. Since
$\mathrm{NoDead}(G)$ holds, there exists $u$ such that $(v,u) \in E$.
Therefore, if a trace $\pi$ visits $v$, the trace can be extended by
$u$. A trace can only be maximal (i.e., cannot be extended) when it
ends at a node in $V_f$ or at a node with no outgoing edges. Since
every non-exit node has outgoing edges, all maximal traces end at exit
nodes.
\end{proof}

\begin{lemma}[Router Shape Soundness]
\label{lem:router_shape}
If $\mathrm{RouterShape}(G)$ holds, then every outgoing transition from
a router node is explicitly labeled as conditional.
\end{lemma}
\begin{proof}
Direct from the predicate definition: the check iterates all edges
$(v,u) \in E$ where $\kappa_V(v) = \textsc{router}$ and verifies
$\kappa_E(v,u) = \textsc{conditional}$. Passing the check establishes
the universal quantification.
\end{proof}

\begin{lemma}[Human Gate Soundness]
\label{lem:human_gate}
If $\mathrm{HumanGate}(G)$ holds, then the graph contains at least one
node classified as a human-in-the-loop step.
\end{lemma}
\begin{proof}
The check searches for $v \in V$ with $\kappa_V(v) = \textsc{human}$.
Passing the check certifies existence.
\end{proof}

\begin{lemma}[Tool Declaration Soundness]
\label{lem:tool_decl}
If $\mathrm{ToolDecl}(G)$ holds, then every tool node explicitly
declares its tool set.
\end{lemma}
\begin{proof}
For each $v$ with $\kappa_V(v) = \textsc{tool}$, the check verifies
$T(v) \neq \emptyset$. Passing the check establishes the universal
statement.
\end{proof}

\begin{lemma}[Reverse Reachability Soundness]
\label{lem:reverse_reach}
If $\mathrm{ExitReachAll}(G)$ holds (i.e.,
$\mathrm{Reach}(v_0) \subseteq \mathrm{RevReach}(V_f)$), then for every
node~$v$ reachable from~$v_0$, there exists a directed path from~$v$ to
some exit node~$v_f \in V_f$.
\end{lemma}
\begin{proof}
The check computes $\mathrm{RevReach}(V_f)$ by BFS on the reverse
adjacency starting from all nodes in~$V_f$. By the correctness of BFS,
$v \in \mathrm{RevReach}(V_f)$ implies the existence of a path
$v \to \cdots \to v_f$ for some $v_f \in V_f$. The check passes only when
$\mathrm{Reach}(v_0) \subseteq \mathrm{RevReach}(V_f)$, so the result
holds for every reachable node.
\end{proof}

\begin{lemma}[Human Gate Coverage Soundness]
\label{lem:human_gate_cov}
If $\mathrm{HumanGateCov}(G, S)$ holds for a set~$S$ of sensitive tool
names, then every path from~$v_0$ to a node invoking a tool in~$S$
passes through at least one \textsc{human} node.
\end{lemma}
\begin{proof}
The check constructs a modified graph~$G'$ by removing all
\textsc{human}-typed nodes and their incident edges, then computes
$\mathrm{Reach}_{G'}(v_0)$.  If no sensitive tool node is in
$\mathrm{Reach}_{G'}(v_0)$, then in the original graph~$G$ every path
from~$v_0$ to such a node must traverse a removed \textsc{human} node.
Conversely, if a sensitive tool node is reachable in~$G'$, the
check fails and reports the human-free path.
\end{proof}

\begin{lemma}[Static Temporal Soundness]
\label{lem:static_temporal}
For a compiled monitor rule with DFA $(Q, q_0, \delta, F_{\mathrm{viol}})$
and an agent workflow graph~$G$, if the graph~$\times$~DFA product
construction reports no violation, then no execution trace of~$G$
violates the temporal property.
\end{lemma}
\begin{proof}
The product BFS explores all reachable states in
$V \times Q$.  Each product state $(v, q)$ represents being at graph
node~$v$ with the DFA in state~$q$.  For each successor~$u$ of~$v$ in~$G$,
the DFA transitions via $q' = \delta(q, \sigma(v))$ where $\sigma(v)$ is the
event symbol for node~$v$, and the product successor $(u, q')$ is enqueued.
If no reachable product state has $q' \in F_{\mathrm{viol}}$, then no
execution trace --- which corresponds to a path through product states ---
can reach a violation state.  This is an over-approximation: the product
explores all \emph{graph paths}, which is a superset of actual runtime traces
(since runtime traces depend on LLM decisions at router nodes).  Hence the
analysis is sound but potentially conservative.
\end{proof}

\begin{lemma}[Temporal Monitor Soundness]
\label{lem:monitor}
For each of the three base DFA patterns (forbidden, implication-future,
until), if the compiled monitor reports no violation on a finite trace
$\pi$, then $\pi$ satisfies the corresponding temporal formula.
\end{lemma}
\begin{proof}
The proof proceeds by cases.

\emph{Forbidden ($\mathbf{G}\,\neg a$):} The DFA has two states:
$s_0$ (safe) and $s_1$ (violated, absorbing). The transition
function maps $s_0 \xrightarrow{a} s_1$ and
$s_0 \xrightarrow{\neg a} s_0$, with $s_1$ absorbing.
If the monitor never enters $s_1$, then $a$ was false at every step,
establishing $\mathbf{G}\,\neg a$.

\emph{Implication-future ($a \to \mathbf{F}\,b$):} The DFA has three
states: $s_0$ (idle), $s_1$ (waiting for $b$), and $s_2$ (violated,
absorbing). Transition: $s_0 \xrightarrow{a \wedge \neg b} s_1$;
$s_1 \xrightarrow{b} s_0$; $s_1 \xrightarrow{a \wedge \neg b} s_2$.
If the monitor never enters $s_2$, then every occurrence of $a$ is
followed by $b$ before $a$ recurs.

\emph{Until ($a\,\mathbf{U}\,b$):} The DFA has three states: $s_0$
(waiting), $s_1$ (satisfied, absorbing), $s_2$ (violated, absorbing).
Transition: $s_0 \xrightarrow{b} s_1$;
$s_0 \xrightarrow{a \wedge \neg b} s_0$;
$s_0 \xrightarrow{\neg a \wedge \neg b} s_2$.
If the monitor never enters $s_2$, then $a$ held on every step until
$b$ occurred, establishing $a\,\mathbf{U}\,b$.

In each case the DFA state invariant is maintained by induction on
the trace prefix.
\end{proof}

\section{Temporal DSL grammar}
\label{app:grammar}

The full BNF grammar for the temporal policy DSL:

\begin{verbatim}
<rule>       ::= <forbidden>
               | <impl_future>
               | <until>
               | <bounded>
               | <chain>
               | <conjunction>
               | <disjunction>

<forbidden>  ::= "G" "!" <atom>

<impl_future>::= <atom> "->" "F" <atom>

<until>      ::= <atom> "U" <atom>

<bounded>    ::= <atom> "->" "F[<=" <int> "]" <atom>

<chain>      ::= <atom> ("->" "F" <atom>){2,}

<conjunction>::= "(" <rule> ")" "AND" "(" <rule> ")"

<disjunction>::= "(" <rule> ")" "OR" "(" <rule> ")"

<atom>       ::= "tool:" <name>
               | "action:" <name>
               | "decision:" <name>
               | <tag_name>

<name>       ::= [a-zA-Z_][a-zA-Z0-9_]*
<tag_name>   ::= [a-zA-Z_][a-zA-Z0-9_]*
<int>        ::= [1-9][0-9]*
\end{verbatim}

\paragraph{Predicate matching.}
Atomic predicates are matched against event dictionaries:
\texttt{tool:X} matches when the event's \texttt{tool\_name} field
equals \texttt{X}; \texttt{action:X} matches \texttt{action\_type};
\texttt{decision:X} matches the \texttt{decision} field; bare names
match membership in the event's \texttt{tags} set.

\end{document}